\documentclass[submission,copyright,creativecommons]{eptcs}
\providecommand{\event}{$V^2$CPS 2016} 
\usepackage{breakurl}             

\usepackage{amsmath, amssymb}
\usepackage{graphicx} 
\usepackage{here} 

\usepackage{theorem}
\theorembodyfont{\rmfamily}
\newtheorem{thm}{Theorem}
\newtheorem{proof}{Proof}
\newtheorem{lem}{Lemma}
\newtheorem{define}{Definition}

\def \qed{\hfill $\Box$} 

\title{Output Feedback Controller Design with Symbolic Observers for Cyber-physical Systems}
\author{
Masashi Mizoguchi \qquad\qquad Toshimitsu Ushio
\institute{Graduate School of Engineering Science\\
Osaka University\\
Osaka, Japan}
\email{\quad mizoguchi@hopf.sys.es.osaka-u.ac.jp \quad\qquad ushio@sys.es.osaka-u.ac.jp}
}
\def\titlerunning{Output Feedback Controller Design with Symbolic Observers for Cyber-physical Systems}
\def\authorrunning{Masashi Mizoguchi \& Toshimitsu Ushio}

\begin{document}
\maketitle

\begin{abstract}
In this paper, we design a symbolic output feedback controller of a cyber-physical system (CPS).
The physical plant is modeled by an infinite transition system.
We consider the situation that a finite abstracted system of the physical plant, called a c-abstracted system, is given.
There exists an approximate alternating simulation relation from the c-abstracted system to the physical plant.
A desired behavior of the c-abstracted system is also given, and we have a symbolic state feedback controller of the physical plant.
We consider the case where some states of the plant are not measured.
Then, to estimate the states with abstracted outputs measured by sensors, we introduce a finite abstracted system of the physical plant, called an o-abstracted system, such that there exists an approximate simulation relation.
The relation guarantees that an observer designed based on the state of the o-abstracted system estimates the current state of the plant.
We construct a symbolic output feedback controller by composing these systems.
By a relation-based approach, we proved that the controlled system approximately exhibits the desired behavior.    

\end{abstract}

\section{Introduction}

Over the past few decades, the control of hybrid systems has been studied \cite{goebel2012hybrid}.
It is a main issue in hybrid systems to find an algorithmic procedure for designing a finite symbolic controller.  
The physical plant has real-valued variables, and its model is an infinite-state system that has too many uncertainties.
Then, finite state abstraction is introduced in verification and synthesis problems of hybrid systems \cite{Tabuada:2009:VCH:1717907}.
The behavior of a finite system can be regarded as a process.
A simulation / bisimulation relation is a key notion that evaluates correctness between two processes \cite{Milner:1989:CC:534666}.
If there exists a simulation / bisimulation relation, the abstracted plant correctly describes the behavior of the physical plant.
But, these relations are often too restrictive in terms of the abstraction because the state set of a hybrid system is infinite \cite{871304}.
Recently, approximate abstraction is considered with approximate simulation / bisimulation relations.
These relations are evaluated by Lyapunov-like functions called simulation / bisimulation functions \cite{4200859, Girard2011568, 5342460, Pola20082508, 4610035}.

Another key concept between two processes is an alternating simulation relation proposed in \cite{Alur98alternatingrefinement}.
This notion is used not only in multi-agent systems and game automata but also to describe a relationship between a feedback controller and a plant.
The control problem of a finite system can be solved by exact alternating simulation relations \cite{Tabuada:2009:VCH:1717907}.
Approximated alternating simulation relations are also introduced as well as simulation / bisimulation relations to consider abstraction \cite{4200859}.

A cyber-physical system (CPS) contains communication networks, which cause disturbances and noises such as data dropouts.
A control performance is often degraded by them, so robustness of the CPS is important.
An approach to design of a symbolic controller under the existence of disturbances is shown in \cite{Bloem:2010:RPL:2144310.2144356, 5351139}, and approximated relations are used in \cite{pola2009symbolic}. 
It is shown that the input-output dynamical stability (IODS) is preserved under the abstraction if there exists an approximated alternating simulation relation \cite{6760490, rungger2014abstracting, 2013arXiv1310.5199R, 6882838}.

On the other hand, the design methods of the output feedback controller have been proposed.
There are a lot of approaches such as game strategies and specification-based estimators \cite{chatterjee2007algorithms, Ehlers:2015:ERS:2728606.2728626, DBLP:journals/tac/GhaemiV14}.
Especially, finite state abstraction is used in \cite{7039992, 6199967, 6882777}.
But, these are not based on (alternating) simulation relations.  
Relation-based approaches can be seen in \cite{cassandras2009introduction, 110009640556} where an observer-based control problem for discrete event systems is considered.  
These approaches are based on the exact alternating simulation / bisimulation relations.
The output feedback controller based on approximated relations is designed in \cite{mizoguchi2015observer} under the assumption that the relation is based on distance between the state of the physical plant and the state of the abstracted plant.

In this paper, we propose a symbolic output feedback controller without introducing the distance.
It is shown that there exists an approximate contractive alternating simulation relation from the proposed symbolic controller to the physical plant.
Since the approximate contractive alternating simulation relation has the contraction property, the abstraction error between the physical plant and the symbolic controller does not diverge.
Thus, the plant controlled by the proposed output feedback controller approximately exhibits the desired behavior.

The rest of this paper is organized as follows.
In Section \ref{Sec_preliminary}, we define a system as a transition system, and introduce notions of approximated relations.
Moreover, we review an approach to design of a symbolic state feedback controller.
In Section \ref{Sec_observer}, we design a symbolic observer based on approximated relations. 
In Section \ref{Sec_output_feedback_control_system}, we construct an output feedback controller. 
It is shown that the controlled plant by the controller approximately exhibits the desired behavior.
The proof is shown in Section \ref{Sec_proof_of_main_results}.
In Section \ref{Sec_illustrative_example}, we consider an example to demonstrate how the proposed controller works.

\section{Preliminaries}
\label{Sec_preliminary}

\subsection{Simulation Relations}

In this subsection, we review several fundamental notions for transition systems \cite{ rungger2014abstracting, 2013arXiv1310.5199R}.

\begin{define}
\label{def_system}
A system $S$ is a tuple $(X, X_0, U, r)$, where:
\begin{itemize}
\item $X$ is a set of states;
\item $X_0 \subseteq X$ is a set of initial states;
\item $U$ is a set of inputs;
\item $r: X \times U \to 2^{X}$ is a transition map.
\end{itemize}
For any $x \in X$, let $U(x) = \{ u \in U \ | \ r(x, u) \neq \emptyset \}$. 
\end{define}

Let $S_1=(X_1, X_{10}, U_1, r_1)$ and $S_2=(X_2, X_{20}, U_2, r_2)$ be two systems.
For a relation $R \subseteq X_1 \times X_2 \times U_1 \times U_2$ over the state sets $X_1, X_2$ and the input sets $U_1, U_2$, denoted by $R_X \subseteq X_1 \times X_2$ is a projection of $R$ to the state sets $X_1, X_2$ defined as follows:
\begin{align*}
R_X = \{ (x_1, x_2) \in X_1 \times X_2 \ | \ \exists u_1 \in U_1, \exists u_2 \in U_2: (x_1, x_2, u_1, u_2) \in R \}.
\end{align*}


\begin{define}
\label{def_acSR}
Let $S_1=(X_1, X_{10}, U_1, r_1)$ and $S_2=(X_2, X_{20}, U_2, r_2)$ be two systems,
let $\kappa, \lambda \in \mathbb{R}_{\ge 0}$, $\beta \in [0, 1[$ be some parameters, and consider a map ${\sf d}: U_1 \times U_2 \to \mathbb{R}_{\ge 0}$.
We call a parameterized (by $\epsilon \in [\kappa, \infty[$) relation $R(\epsilon) \subseteq X_1 \times X_2 \times U_1 \times U_2$ a $\kappa$-approximate $(\beta, \lambda)$-contractive simulation relation ($(\kappa, \beta, \lambda)$-acSR) from $S_1$ to $S_2$ with ${\sf d}$ if $R(\epsilon) \subseteq R(\epsilon')$ holds for all $\epsilon \le \epsilon'$ and the following two conditions hold for all $\epsilon \in [\kappa, \infty[$.
\begin{enumerate}
\item $\forall x_{10} \in X_{10}, \exists x_{20} \in X_{20}: (x_{10}, x_{20}) \in R_X(\kappa)$;
\item $\forall (x_1, x_2) \in R_X(\epsilon)$, $\forall u_1 \in U_1(x_1)$, $\exists u_2 \in U_2(x_2)$:
\begin{align*}
(x_1, x_2, u_1, u_2) \in R(\epsilon) \ \land \ \forall x'_1 \in r_1(x_1, u_1), \exists x'_2 \in r_2(x_2, u_2): (x'_1, x'_2) \in R_X(\kappa + \beta \epsilon + \lambda {\sf d}({u}_1, u_2)).
\end{align*}
\end{enumerate}
We call $R(\epsilon)$ a simulation relation (SR) from $S_1$ to $S_2$ if $R(\epsilon)$ is a $(0,0,0)$-acSR from $S_1$ to $S_2$.
\end{define}

\begin{define}
\label{def_acASR}
Let $S_1=(X_1, X_{10}, U_1, r_1)$ and $S_2=(X_2, X_{20}, U_2, r_2)$ be two systems,
let $\kappa, \lambda \in \mathbb{R}_{\ge 0}$, $\beta \in [0, 1[$ be some parameters, and consider a map ${\sf d}: U_1 \times U_2 \to \mathbb{R}_{\ge 0}$.
We call a parameterized (by $\epsilon \in [\kappa, \infty[$) relation $R(\epsilon) \subseteq X_1 \times X_2 \times U_1 \times U_2$ a $\kappa$-approximate $(\beta, \lambda)$-contractive alternating simulation relation ($(\kappa, \beta, \lambda)$-acASR) from $S_1$ to $S_2$ with ${\sf d}$ if $R(\epsilon) \subseteq R(\epsilon')$ holds for all $\epsilon \le \epsilon'$ and the following two conditions hold for all $\epsilon \in [\kappa, \infty[$.
\begin{enumerate}
\item $\forall x_{10} \in X_{10}, \exists x_{20} \in X_{20}: (x_{10}, x_{20}) \in R_X(\kappa)$;
\item $\forall (x_1, x_2) \in R_X(\epsilon)$, $\forall u_1 \in U_1(x_1)$, $\exists u_2 \in U_2(x_2)$:
\begin{align*}
(x_1, x_2, u_1, u_2) \in R(\epsilon) \ \land \ \forall x'_2 \in r_2(x_2, u_2), \exists x'_1 \in r_1(x_1, u_1): (x'_1, x'_2) \in R_X(\kappa + \beta \epsilon + \lambda {\sf d}(u_1, u_2)).
\end{align*}
\end{enumerate}
We call $R(\epsilon)$ an alternating simulation relation (ASR) from $S_1$ to $S_2$ if $R(\epsilon)$ is a $(0,0,0)$-acASR from $S_1$ to $S_2$.
\end{define}

\begin{define}
\label{def_composed_system}
Let $S_1=(X_1, X_{10}, U_1, r_1)$ and $S_2=(X_2, X_{20}, U_2, r_2)$ be two systems, and let $R \subseteq X_1 \times X_2 \times U_1 \times U_2$ be a relation.
We define the composition of $S_1$ and $S_2$ with respect to $R$, denoted by $S := S_1 \times_{R} S_2 = (X, X_0, U, r)$ where:
\begin{itemize}
\item $X = X_1 \times X_2$;
\item $X_0 = (X_{10} \times X_{20}) \cap R_X$;
\item $U = U_1 \times U_2$;
\item $r: X \times U \to 2^{X}$ is defined as follows: $(x'_1, x'_2) \in r((x_1, x_2), (u_1, u_2))$ if and only if  
\begin{align*}
x'_1 \in r_1(x_1, u_1) \ \land \ x'_2 \in r_2(x_2, u_2) \ \land \ (x_1, x_2, u_1, u_2) \in R \ \land \ (x'_1, x'_2) \in R_X.
\end{align*}
\end{itemize}
If $R(\epsilon)$ is a $(\kappa, \beta, \lambda)$-acSR or a $(\kappa, \beta, \lambda)$-acASR from $S_1$ to $S_2$ with $\sf d$, we replace the above definitions of $X_0$ and $r$ with the following conditions:
\begin{itemize}
\item $X_0 = (X_{10} \times X_{20}) \cap R_X(\kappa)$;
\item $r: X \times U \to 2^{X}$ is defined as follows: $(x'_1, x'_2) \in r((x_1, x_2), (u_1, u_2))$ if and only if  
\begin{align*}
x'_1 \in r_1(x_1, u_1) \ \land \ x'_2 \in r_2(x_2, u_2) \ &\land \ (x_1, x_2, u_1, u_2) \in R(e(x_1, x_2)) \\
&\land \ (x'_1, x'_2) \in R_X(\kappa + \beta e(x_1, x_2) + \lambda {\sf d}(u_1, u_2)),
\end{align*}
where $e(x_1, x_2) := \inf\{ \epsilon \in \mathbb{R}_{\ge 0} \ | \ (x_1, x_2) \in R_X(\epsilon) \}$. 
\end{itemize}    
\end{define}

\subsection{State Feedback}
\label{Sec_state_feedback}

In this subsection, we review the state feedback control \cite{ rungger2014abstracting, 2013arXiv1310.5199R}.
A physical plant to be controlled, denoted by $S=(X, X_0, U, r)$, is a discrete time system, and its state set $X$ is a Euclidean space.
Thus, $S$ is an infinite transition system. 
In order to design a digital controller, we introduce a finite abstracted model of the plant that is implemented in the cyber space.
We consider an abstracted system $\hat{S}=(\hat{X}, \hat{X}_0, \hat{U}, \hat{r})$ of the plant $S$, called a c-abstracted system, such that there exists a $(\kappa, \beta, \lambda)$-acASR ${R}(\epsilon) \subseteq \hat{X} \times {X} \times \hat{U} \times {U}$ from $\hat{S}$ to $S$.
A desired behavior of the plant is described by $\hat{S}_C=(\hat{X}_C, \hat{X}_{C0}, \hat{U}_C, \hat{r}_C)$ based on $\hat{S}$ in such a way that there exists an ASR $\hat{R}_C \subseteq \hat{X}_C \times \hat{X} \times \hat{U}_C \times \hat{U}$ from $\hat{S}_C$ to $\hat{S}$.
Then, the following theorem shows the existence of a symbolic state feedback controller \cite{rungger2014abstracting, 2013arXiv1310.5199R}.

\begin{thm}
\label{tabuada}
Let $\hat{S}_C=(\hat{X}_C, \hat{X}_{C0}, \hat{U}_C, \hat{r}_C)$, $\hat{S}=(\hat{X}, \hat{X}_0, \hat{U}, \hat{r})$, and $S=(X, X_0, U, r)$ be systems, let $\kappa, \lambda \in \mathbb{R}_{\ge 0}$, $\beta \in [0, 1[$ be some parameters, and consider a map ${\sf d}: U \times \hat{U} \to \mathbb{R}_{\ge 0}$.
Assume that there exist an ASR $\hat{R}_C \subseteq \hat{X}_C \times \hat{X} \times \hat{U}_C \times \hat{U}$ from $\hat{S}_C$ to $\hat{S}$ and a $(\kappa, \beta, \lambda)$-acASR $R(\epsilon) \subseteq \hat{X} \times X \times \hat{U} \times U$ from $\hat{S}$ to $S$ with ${\sf d}$.
Then, the following relation $R_C(\epsilon) \subseteq (\hat{X}_C \times \hat{X}) \times X \times (\hat{U}_C \times \hat{U}) \times U$ is a $(\kappa, \beta, \lambda)$-acASR from $S_C := \hat{S}_C \times_{\hat{R}_C} \hat{S}$ to $S$ with ${\sf d}_C((\hat{u}_C, \hat{u}), u) = {\sf d}(\hat{u}, u)$.
\begin{align}
\label{eq_R_C(epsilon)}
R_C(\epsilon) = \{ ((\hat{x}_C, \hat{x}), x, (\hat{u}_C, \hat{u}), u) \ | \ (\hat{x}, x, \hat{u}, u) \in {R}(\epsilon) \ \land \ (\hat{x}_C, \hat{x}) \in \hat{R}_{CX} \}.
\end{align}
\end{thm}

Theorem \ref{tabuada} implies that $S_C$, which is the composition of the abstracted systems $\hat{S}_C$ and $\hat{S}$, is a controller of $S$. 
In the case where the state $x$ of $S$ is fully observed, we can determine the state $\hat{x}$ of $\hat{S}$ by the relation $R(\epsilon)$.
Then, the state $\hat{x}_C$ of $\hat{S}_C$ is determined by the relation $\hat{R}_C$.
The controller determines a control input $\hat{u}_C \in \hat{U}_C$ such that $\hat{r}_C(\hat{x}_C, \hat{u}_C) \neq \emptyset$.

However, in general, all states are not always measured.
Moreover, the output value is abstracted by the resolution of sensors.
Then, in the next section, we introduce a symbolic observer based on abstracted outputs.   

\section{Observer Design}
\label{Sec_observer}

Let $Y$ be a set of outputs, and $H: X \to Y$ be an output map of the plant $S=(X, X_0, U, r)$.
We call $(S, Y, H) = (X, X_0, U, r, Y, H)$ a plant with outputs.
The observer introduced in this section is based on the concept of observers for discrete event systems \cite{cassandras2009introduction}.
In order that the observer always estimates the current state of the plant, it must simulate any behavior of the plant. 
Thus, we introduce an abstracted system $\check{S}=(\check{X}, \check{X}_0, \check{U}, \check{r})$ of the plant $S$, called an o-abstracted system, such that there exists an acSR $\check{R}(\epsilon)$ from $S$ to $\check{S}$.
Note that $\check{S} \neq \hat{S}$ in general. 
The acSR $\check{R}(\epsilon)$ is based on measured outputs of the physical plant and satisfies the following condition:
\begin{align}
\forall x_1 \in X, \forall x_2 \in X, \forall \check{x} \in \check{X}: (x_1, \check{x}) \in \check{R}_X(\epsilon) \ \land \ (x_2, \check{x}) \in \check{R}_X(\epsilon) \Rightarrow H({x}_1) = H({x}_2).
\end{align}

\begin{define}
\label{def_observer}
Let $(S, Y, H)$ be a plant with outputs and $\check{S}$ be an o-abstracted system, let $\kappa', \lambda' \in \mathbb{R}_{\ge 0}$, $\beta' \in [0, 1[$ be some parameters, and consider a map $\check{\sf d}: X \times \check{X} \to \mathbb{R}_{\ge 0}$.
There exists a $(\kappa', \beta', \lambda')$-acSR $\check{R}(\epsilon) \subseteq X \times \check{X} \times U \times \check{U}$ from $S$ to $\check{S}$ with $\check{\sf d}$.  
Then, we define a system $\tilde{S}=(\tilde{X}, \tilde{X}_0, \tilde{U}, \tilde{r})$ where:

\begin{itemize}
\item ${\tilde{X}} = 2^{\check{X}} \setminus \{ \emptyset \}$;
\item ${\tilde{X}_0} = 2^{\check{X}_0} \setminus \{ \emptyset \} \subseteq {\tilde{X}}$;
\item $\tilde{U} = \check{U}$;
\item $\tilde{r}: \tilde{X} \times \tilde{U} \to 2^{\tilde{X}}$ is defined as follows:
\begin{align*}
\tilde{r}(\tilde{x}, \check{u})= \{ \tilde{x}' \in \tilde{X} \ | \ &\forall \check{x}' \in \tilde{x}', \exists \check{x} \in \tilde{x}, \exists x \in X, \exists u \in U, \exists x' \in r(x, u), \exists \epsilon \in [\kappa', \infty[:\\
&\check{x}' \in \check{r}(\check{x}, \check{u}) \land (x, \check{x}, u, \check{u}) \in \check{R}(\epsilon) \land (x', \check{x}') \in \check{R}_X(\kappa' + \beta' \epsilon + \lambda' \check{\sf d}(u, \check{u})) \},
\end{align*}
where $\check{R}_X(\epsilon) \subseteq X \times \check{X}$ is the projection of a relation induced by $\check{R}(\epsilon)$ defined as follows:
\begin{align*}
\check{R}_X(\epsilon) = \{ (x, \check{x}) \in X \times \check{X} \ | \ \exists u \in U, \exists \check{u} \in \check{U}: (x, \check{x}, u, \check{u}) \in \check{R}(\epsilon) \}.
\end{align*}
\end{itemize}
We call $\tilde{S}$ an observer of $S$ induced by $\check{S}$.
\end{define}

\begin{thm}
\label{thm_obs_acSR}
Let $(S, Y, H)$ be a plant with outputs, $\check{S}$ is an o-abstracted system, and there exists a $(\kappa', \beta', \lambda')$-acSR $\check{R}(\epsilon) \subseteq X \times \check{X} \times U \times \check{U}$ from $S$ to $\check{S}$ with $\check{\sf d}$ for some $\kappa', \lambda' \in \mathbb{R}_{\ge 0}$, $\beta' \in [0, 1[$, and a map $\check{\sf d}: X \times \check{X} \to \mathbb{R}_{\ge 0}$.
The observer $\tilde{S}$ is induced by $\check{S}$ as in Definition \ref{def_observer}.
Then, the following relation $R'(\epsilon) \subseteq X \times \tilde{X} \times U \times \tilde{U}$ is a $(\kappa', \beta', \lambda')$-acSR from $(S, Y, H)$ to $\tilde{S}$ with $\check{\sf d}$:
\begin{align}
\label{def_R'(epsilon)}
R'(\epsilon) = \{ (x, \tilde{x}, u, \check{u}) \ | \ \exists \check{x} \in \tilde{x}: (\check{x}, {x}, u, \check{u}) \in \check{R}(\epsilon) \}.
\end{align}
\end{thm}

\begin{proof}
We will show that ${R}'(\epsilon)$ satisfies the conditions of a $(\kappa', \beta', \lambda')$-acSR from $S$ to $\tilde{S}$ with $\check{\sf d}$.
\begin{enumerate}
\item Consider any ${x}_{0} \in {X}_{0}$.
Let $H(x_0)=y_0$. 
We consider the following state $\tilde{x}_0$ of $\tilde{S}$:
\begin{align*}
{\tilde{x}_0} = \{ \check{x}_0 \in \check{X}_0 \ | \ \exists {x}_{p0} \in {X}_{0}: H(x_{p0}) = y_0 \ \land \  ({x}_{p0}, \check{x}_0) \in \check{R}_{X}(\kappa') \}.
\end{align*}
Note that by the $(\kappa', \beta', \lambda')$-acSR $\check{R}(\epsilon)$, ${\tilde{x}_0}$ is a non-empty set. 
Recall that ${\tilde{X}_0}=2^{\check{X}_0} \setminus \{ \emptyset \}$.
Then, we have ${\tilde{x}_0} \in { \tilde{X}_0}$. 
By the definition of ${R}'(\epsilon)$, $({x}_{0}, \tilde{x}_0) \in {R}'_X(\kappa')$ holds. 
\item First, consider any $({x}, \tilde{x}) \in R'_X(\epsilon)$.
We have
\begin{align*}
\exists \check{x} \in \tilde{x}: (x, \check{x}) \in \check{R}_X(\epsilon).
\end{align*}
Choose any $u \in U(x)$. 
By the $(\kappa', \beta', \lambda')$-acSR $\check{R}(\epsilon)$, there exists $\check{u} \in \tilde{U}(\tilde{x})$ such that $(x, \check{x}, u, \check{u}) \in \check{R}(\epsilon)$. Now, we have $({x}, \tilde{x}, u, \check{u}) \in {R}'(\epsilon)$.

Next, consider any ${x}' \in {r}(x, u)$.
Let $H(x') = y'$.
We consider the following state $\tilde{x}'$ of $\tilde{S}$:
\begin{align*}
\tilde{x}' = \{ \check{x}' \in \check{X} \ | \ &\exists \check{x} \in \tilde{x}, \exists x_p \in X, \exists x'_p \in r(x_p, u):\\
&\check{x}' \in \check{r}(\check{x}, \check{u}) \land H(x'_p)=y' \land (x_p, \check{x}, u, \check{u}) \in \check{R}(\epsilon) \land (x'_p, \check{x}') \in \check{R}_X(\kappa' + \beta' \epsilon + \lambda' \check{\sf d}(u, \check{u})) \}.
\end{align*}
Note that by the $(\kappa', \beta', \lambda')$-acSR $\check{R}(\epsilon)$, $\tilde{x}'$ is a non-empty set. 
Moreover, we have
\begin{align*}
\exists \check{x}' \in \tilde{x}': (x', \check{x}') \in \check{R}_X(\kappa'+\beta'\epsilon+\lambda' \check{\sf d}(u, \check{u})).
\end{align*}
By the definition of $\tilde{r}$, $\tilde{x}' \in \tilde{r}(\tilde{x}, \check{u})$ holds.
Thus, by the definition of ${R}'(\epsilon)$, $(x', \tilde{x}') \in {{R}'_{X}}(\kappa'+\beta'\epsilon+\lambda' \check{\sf d}(u, \check{u}))$ holds. \qed
\end{enumerate}
\end{proof}

Theorem \ref{thm_obs_acSR} implies that for the current state $\tilde{x}$ of the observer, there always exists $\check{x} \in \tilde{x}$ such that $\check{x}$ is an o-abstracted state of the current state $x$ of the plant.
Thus, $\tilde{S}$ lists up all candidates of $\check{x}$.

\section{Output Feedback Control System}
\label{Sec_output_feedback_control_system}

In this section, we construct an output feedback system with the observer defined in the previous section.

\begin{define}
\label{def_bf_Shat}
We define a system ${\bf \hat{S}}=({\bf \hat{X}}, {\bf \hat{X}_0}, {\bf \hat{U}}, {\bf \hat{r}})$ induced by $\hat{S}=(\hat{X}, \hat{X}_0, \hat{U}, \hat{r})$ where:
\begin{itemize}
\item ${\bf \hat{X}} = 2^{\hat{X}} \setminus \{ \emptyset \}$;
\item ${\bf \hat{X}_0} = 2^{\hat{X}_0} \setminus \{ \emptyset \} \subseteq {\bf \hat{X}}$;
\item ${\bf \hat{U}} = \hat{U}$;
\item ${\bf \hat{r}}: {\bf \hat{X}} \times {\bf \hat{U}} \to 2^{\bf \hat{X}}$ is defined as follows:
\begin{align*}
{\bf \hat{r}}({\bf \hat{x}}, \hat{u}) = \begin{cases}
2^{\bigcup_{\hat{x} \in {\bf \hat{x}}}{\hat{r}(\hat{x}, \hat{u})}} \setminus \{ \emptyset \} &\mbox{if \ } \forall \hat{x} \in {\bf \hat{x}}: \hat{r}(\hat{x}, \hat{u}) \neq \emptyset, \\
\emptyset &\mbox{otherwise}.
\end{cases}
\end{align*}
\end{itemize}
\end{define}

Intuitively, each $\hat{x} \in {\bf \hat{x}}$ corresponds to (at least) one candidate listed up by the observer.
Recall that $\tilde{S}$ is induced by $\check{S}$ that is different from $\hat{S}$.
Then, we consider ${\bf \hat{x}}$ such that each c-abstracted state $\hat{x} \in {\bf \hat{x}}$ corresponds to (at least) one o-abstracted state $\check{x} \in \tilde{x}$.  
The transition map $\bf \hat{r}$ is defined only when $\hat{r}$ is defined for all $\hat{x} \in {\bf \hat{x}}$.

We define $\bf \hat{S}_C$ induced by $\hat{S}_C$ as well as $\bf \hat{S}$ induced by $\hat{S}$.

\begin{define}
\label{def_bf_S_Chat}
We define a system ${\bf \hat{S}_C}=({\bf \hat{X}_C}, {\bf \hat{X}_{C0}}, {\bf \hat{U}_C}, {\bf \hat{r}_C})$ induced by $\hat{S}_C=(\hat{X}_C, \hat{X}_{C0}, \hat{U}_C, \hat{r}_C)$ where:
\begin{itemize}
\item ${\bf \hat{X}_C} = 2^{\hat{X}_C} \setminus \{ \emptyset \}$;
\item ${\bf \hat{X}_{C0}} = 2^{\hat{X}_{C0}} \setminus \{ \emptyset \} \subseteq {\bf \hat{X}_C}$;
\item ${\bf \hat{U}_C} = \hat{U}_C$;
\item ${\bf \hat{r}_C}: {\bf \hat{X}_C} \times {\bf \hat{U}_C} \to 2^{\bf \hat{X}_C}$ is defined as follows:
\begin{align*}
{\bf \hat{r}_C}({\bf \hat{x}_C}, \hat{u}_C) = \begin{cases}
2^{\bigcup_{\hat{x}_C \in {\bf \hat{x}_C}}{\hat{r}_C(\hat{x}_C, \hat{u}_C)}} \setminus \{ \emptyset \} &\mbox{if \ } \forall \hat{x}_C \in {\bf \hat{x}_C}: \hat{r}_C(\hat{x}_C, \hat{u}_C) \neq \emptyset, \\
\emptyset &\mbox{otherwise}.
\end{cases}
\end{align*}
\end{itemize}
\end{define}

Then, we have the following main theorems.

\begin{thm}
\label{result_bf_R_C(epsilon)}
Consider the same condition as Theorem \ref{tabuada}.
In addition, let $\check{S}=(\check{X}, \check{X}_0, \check{U}, \check{r})$ be an o-abstracted system, and there exists a $(\kappa', \beta', \lambda')$-acSR $\check{R}(\epsilon) \subseteq X \times \check{X} \times U \times \check{U}$ from $S$ to $\check{S}$ with $\check{\sf d}$ for some $\kappa', \lambda' \in \mathbb{R}_{\ge 0}$, $\beta' \in [0, 1[$, and a map $\check{\sf d}: X \times \check{X} \to \mathbb{R}_{\ge 0}$.
Let $\tilde{S}=(\tilde{X}, \tilde{X}_0, \tilde{U}, \tilde{r})$ be an observer induced by $\check{S}$ as in Definition \ref{def_observer}.
Let ${\bf \hat{S}}=({\bf \hat{X}}, {\bf \hat{X}_0}, {\bf \hat{U}}, {\bf \hat{r}})$ and ${\bf \hat{S}_C}=({\bf \hat{X}_C}, {\bf \hat{X}_{C0}}, {\bf \hat{U}_C}, {\bf \hat{r}_C})$ be systems induced by a c-abstracted system $\hat{S}$ and $\hat{S}_C$, respectively as in Definitions \ref{def_bf_Shat} and \ref{def_bf_S_Chat}.
Assume that $\hat{R}_C$ satisfies (\ref{restrictive_R_Chat}), that there exists a map $\bar{\sf d}: \hat{U} \times \check{U} \to \mathbb{R}_{\ge 0}$ satisfying (\ref{assum_map}), and that $R(\epsilon)$ and $\check{R}(\epsilon)$ satisfy (\ref{assum_R(epsilon)}) and (\ref{assum_Rcheck(epsilon)}), respectively.
\begin{align}
\label{restrictive_R_Chat}
&\forall \hat{u}_C \in \hat{U}_C, \exists \hat{u} \in \hat{U}, \forall (\hat{x}_C, \hat{x}) \in \hat{R}_{CX}: \ \hat{u}_C \in \hat{U}_C(\hat{x}_C) \Rightarrow (\hat{x}_C, \hat{x}, \hat{u}_C, \hat{u}) \in \hat{R}_C,\\
\label{assum_map}
&\forall \hat{u} \in \hat{U}, \forall u \in U, \forall \check{u} \in \check{U}: \ \bar{\sf d}(\hat{u}, \check{u}) \ge {\sf d}(\hat{u}, u) + \check{\sf d}(u, \check{u}),\\
\label{assum_R(epsilon)}
&\forall \epsilon \in [\kappa, \infty[, \forall \hat{u} \in \hat{U}, \exists u \in U, \forall (\hat{x}, x) \in R_X(\epsilon): \ \hat{u} \in \hat{U}(\hat{x}) \Rightarrow (\hat{x}, x, \hat{u}, u) \in R(\epsilon),\\
\label{assum_Rcheck(epsilon)}
&\forall \epsilon \in [\kappa, \infty[, \forall {u} \in {U}, \exists \check{u} \in \check{U}, \forall (x, \check{x}) \in \check{R}_X(\epsilon): \ u \in U(x) \Rightarrow ({x}, \check{x}, {u}, \check{u}) \in \check{R}(\epsilon). 
\end{align}

\begin{figure}
\centering
\includegraphics[width=10 cm]{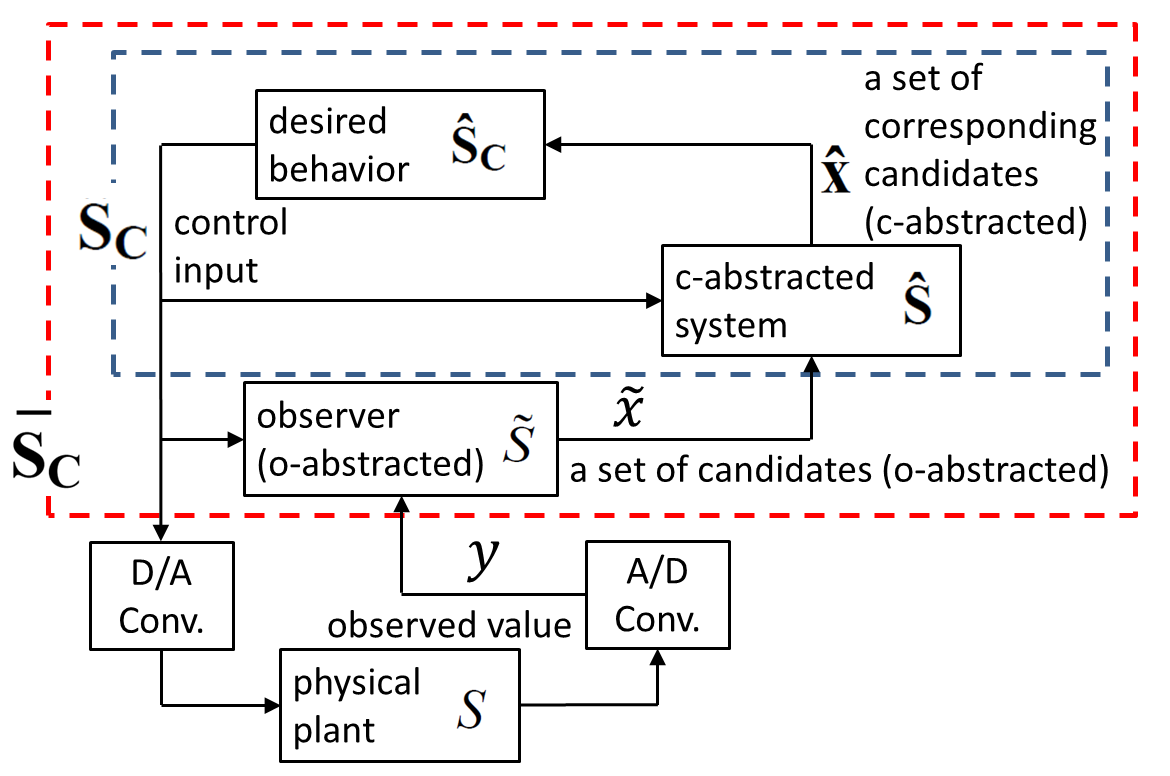}

\caption{The proposed output feedback controller ${\bf \bar{S}_C} = {\bf S_C} \times_{{\bf R_C}(\epsilon)} \tilde{S}$ consisting of $\bf S_C = \hat{S}_C \times_{\hat{R}_C} \hat{S}$ and the observer $\tilde{S}$.}
\label{fig_structure}
\end{figure}

Let $\bf S_C := \hat{S}_C \times_{\hat{R}_C} \hat{S}$.
Then, the following relation ${\bf R_C}(\epsilon) \subseteq ({\bf \hat{X}_C} \times {\bf \hat{X}}) \times \tilde{X} \times ({\bf \hat{U}_C} \times {\bf \hat{U}}) \times \tilde{U}$ is a $(\kappa + \kappa', \max\{\beta,\beta'\}, \max\{\lambda, \lambda'\})$-acASR from $\bf S_C$ to $\tilde{S}$ with ${\sf d_C}((\hat{u}_C, \hat{u}), \check{u}) = \bar{\sf d}(\hat{u}, \check{u})$.
\begin{align}
\label{def_bf_RC(epsilon)}
{\bf R_C}(\epsilon)= \{ (({\bf \hat{x}_C}, {\bf \hat{x}}), \tilde{x}, (\hat{u}_C, \hat{u}), \check{u}) \ | \ ({\bf \hat{x}}, \tilde{x}, \hat{u}, \check{u}) \in {\bf R}(\epsilon) \land ({\bf \hat{x}_C}, {\bf \hat{x}}) \in {\bf \hat{R}}_{{\bf C}X} \},
\end{align}
where the relations ${\bf \hat{R}_C} \subseteq {\bf \hat{X}_C} \times {\bf \hat{X}} \times {\bf \hat{U}_C} \times {\bf \hat{U}}$ and ${\bf R}(\epsilon) \subseteq {\bf \hat{X}} \times \tilde{X} \times {\bf \hat{U}} \times \tilde{U}$ are defined by (\ref{def_bf_R_Chat}) and (\ref{def_bf_R(epsilon)}), respectively.
\begin{align}
\label{def_bf_R_Chat}
&{\bf \hat{R}_C} = \{ ({\bf \hat{x}_C}, {\bf \hat{x}}, \hat{u}_C, \hat{u}) \ | \ \forall \hat{x} \in {\bf \hat{x}}, \exists \hat{x}_C \in {\bf \hat{x}_C}: (\hat{x}_C, \hat{x}, \hat{u}_C, \hat{u}) \in \hat{R}_C \},\\
&{\bf R}(\epsilon) = \{ ({\bf \hat{x}}, \tilde{x}, \hat{u}, \check{u}) \ | \ \exists u \in U, \exists \epsilon_1 \ge \kappa, \exists \epsilon_2 \ge \kappa', \forall \check{x} \in {\tilde{x}}, \exists x \in X, \exists \hat{x} \in {\bf \hat{x}}: \nonumber \\
&\ \ \ \ \ \ \ \ \ \ \ \ \ \ \ \ \ \ \ \ \ \ \ \ \ \ \ \ \ \ \ \ \epsilon_1 + \epsilon_2 = \epsilon \ \land \ (\hat{x}, x, \hat{u}, u) \in R(\epsilon_1) \ \land \ (x, \check{x}, u, \check{u}) \in \check{R}(\epsilon_2) \}. \label{def_bf_R(epsilon)}
\end{align}
\end{thm}

\begin{thm}
\label{result_main}
Consider the same condition as Theorem \ref{result_bf_R_C(epsilon)}.
Let ${\bf \bar{S}_C} := {\bf S_C} \times_{{\bf R_C}(\epsilon)} \tilde{S} = (\bar{\bf X}_{\bf C}, \bar{\bf X}_{\bf C0}, \bar{\bf U}_{\bf C}, \bar{\bf r}_{\bf C})$.
The following relation ${\bf \bar{R}_C}(\epsilon) \subseteq (({\bf \hat{X}_C} \times {\bf \hat{X}}) \times \tilde{X}) \times X \times (({\bf \hat{U}_C} \times {\bf \hat{U}}) \times \tilde{U}) \times U$ is a $(\kappa + \kappa', \max\{ \beta, \beta'\}, \max\{ \lambda, \lambda' \\ \})$-acASR from $\bf \bar{S}_C$ to $({S}, Y, H)$ with $\bar{\sf d}_{\sf C}(((\hat{u}_C, \hat{u}), \check{u}), u) = \bar{\sf d}(\hat{u}, \check{u})$:
\begin{align}
{\bf \bar{R}_C}(\epsilon) = \{((({\bf \hat{x}_C}, {\bf \hat{x}}), \tilde{x}), x, ((\hat{u}_C, \hat{u}), \check{u}), u) \ | \ &\exists \epsilon_1 \ge \kappa, \exists \epsilon_2 \ge \kappa', \exists \hat{x} \in {\bf \hat{x}}: \nonumber \\
&\epsilon_1 + \epsilon_2 = \epsilon \land (\hat{x}, x, \hat{u}, u) \in R(\epsilon_1) \land (x, \tilde{x}, u, \check{u}) \in {R}'(\epsilon_2) \},
\label{def_bf_RCbar(epsilon)}
\end{align}
\end{thm}
where $R'(\epsilon)$ is defined by (\ref{def_R'(epsilon)}).

The definition of ${\bf R_C}(\epsilon)$ corresponds to that of $R_C(\epsilon)$ given by (\ref{eq_R_C(epsilon)}).
Then, Theorem \ref{result_main} is an extension of Theorem \ref{tabuada} to the output feedback control of the physical plant.
The block diagram of the proposed control system is shown in Fig.~\ref{fig_structure}.
When the observer $\tilde{S}$ receives the output $y$, it updates the estimtion $\tilde{x}$ of the current state of the plant. Then, $\bf {S}_C$ updates the state $({\bf \hat{x}_C}, {\bf \hat{x}})$ by the relation ${\bf R_C}(\epsilon)$.
The controller determines a control input $\hat{u}_C \in {\bf \hat{U}_C(\hat{x}_C)}$ such that for any $\hat{x}_C \in {\bf \hat{x}_C}$, $\hat{r}_C(\hat{x}_C, \hat{u}_C) \neq \emptyset$ holds.
This is the key idea of our approach. 
Then, by Theorem \ref{result_main}, the physical plant controlled by the output feedback controller $\bf \bar{S}_{C}$ exhibits a desired behavior approximately in the sense of the acASR.

In the case where $\hat{S}$ and $\check{S}$ are constructed by discretization of the state set $X$ of the plant, we can define the parameterized relations based on the Euclidean distance on $X$ that measures the abstraction error as shown in the illustrative example.  

In the next section, we show the proofs of these theorems.

\section{Proofs of Main Theorems}
\label{Sec_proof_of_main_results}
%
First, we show two lemmas that are needed in the proofs of the main theorems.

\begin{lem}
\label{lem_bf hat{R}_C}
The relation ${\bf \hat{R}_C} \subseteq {\bf \hat{X}_C} \times {\bf \hat{X}} \times {\bf \hat{U}_C} \times {\bf \hat{U}}$ defined by (\ref{def_bf_R_Chat}) is an ASR from $\bf \hat{S}_C$ to $\bf \hat{S}$.
\end{lem}
\begin{proof}
We will show that $\bf \hat{R}_C$ satisfies the conditions of an ASR from $\bf \hat{S}_C$ to $\bf \hat{S}$.
\begin{enumerate}
\item Consider any $\bf \hat{x}_{C0} \in \hat{X}_{C0}$. 
We consider the following state $\bf \hat{x}_0$ of $\bf \hat{S}$:
\begin{align*}
{\bf \hat{x}_0} = \{ \hat{x}_0 \in \hat{X}_0 \ | \ \exists \hat{x}_{C0} \in {\bf \hat{x}_{C0}}: (\hat{x}_{C0}, \hat{x}_0) \in \hat{R}_{CX} \}.
\end{align*}
Since $\hat{R}_C$ is an ASR, ${\bf \hat{x}_0}$ is a non-empty set. 
Recall that ${\bf \hat{X}_0}=2^{\hat{X}_0} \setminus \{ \emptyset \}$, and we have ${\bf \hat{x}_0} \in {\bf \hat{X}_0}$. 
By the definition of $\bf \hat{R}_C$, ${\bf (\hat{x}_{C0}, \hat{x}_0) \in \hat{R}}_{{\bf C}X}$ holds. 
\item First, consider any ${\bf (\hat{x}_{C}, \hat{x}) \in \hat{R}}_{{\bf C}X}$.
Choose any $\hat{u}_C \in {\bf \hat{U}_C}({\bf \hat{x}_C})$.
From (\ref{restrictive_R_Chat}), there exists $\hat{u} \in {\bf \hat{U}}({\bf \hat{x}})$ such that $({\bf \hat{x}_{C}}, {\bf \hat{x}}, \hat{u}_C, \hat{u}) \in {\bf \hat{R}_{C}}$.
By the definition of $\bf \hat{R}_C$, the following condition holds:
\begin{align*}
\forall \hat{x} \in {\bf \hat{x}}, \exists \hat{x}_C \in {\bf \hat{x}_C}: (\hat{x}_C, \hat{x}, \hat{u}_C, \hat{u}) \in \hat{R}_C,
\end{align*}
which implies together with the definition of the ASR that
\begin{align}
\forall \hat{x}' \in \bigcup_{\hat{x} \in {\bf \hat{x}}}{\hat{r}(\hat{x}, \hat{u})}, \ \exists \hat{x}'_C \in \bigcup_{\hat{x}_C \in {\bf \hat{x}_C}}{\hat{r}_C(\hat{x}_C, \hat{u}_C)}: \ (\hat{x}'_C, \hat{x}') \in \hat{R}_{CX}.
\label{eq_proof_1}
\end{align}
Next, consider any ${\bf \hat{x}}' \in {\bf \hat{r}}({\bf \hat{x}}, \hat{u})$.
By the definition of $\bf \hat{r}$, we have
\begin{align*}
{\bf \hat{x}'} \subseteq \bigcup_{\hat{x} \in {\bf \hat{x}}}{\hat{r}(\hat{x}, \hat{u})}.
\end{align*}
By the definition of $\bf \hat{r}_C$, we have
\begin{align*}
{\bf \hat{r}_C}({\bf \hat{x}_C}, \hat{u}_C) =2^{\bigcup_{\hat{x}_C \in {\bf \hat{x}_C}}{\hat{r}_C(\hat{x}_C, \hat{u}_C)}} \setminus \{ \emptyset \}.
\end{align*}
Thus, from (\ref{eq_proof_1}), there always exists ${\bf \hat{x}'_C} \in {\bf \hat{r}_C}({\bf \hat{x}_C}, \hat{u}_C)$ satisfying the following condition:
\begin{align*}
\forall \hat{x}' \in {\bf \hat{x}'}, \ \exists \hat{x}'_C \in {\bf \hat{x}'_C}: \ (\hat{x}'_C, \hat{x}') \in \hat{R}_{CX}. 
\end{align*}
Therefore, by the definition of $\bf \hat{R}_C$, $({\bf \hat{x}'_C}, {\bf \hat{x}'}) \in {\bf \hat{R}}_{{\bf C}X}$ holds. \qed
\end{enumerate}
\end{proof}

Lemma \ref{lem_bf hat{R}_C} shows that $\bf \hat{S}_C$ is a feedback controller of $\bf \hat{S}$.

\begin{lem}
\label{lem_bf_R}
The relation ${\bf R}(\epsilon) \subseteq {\bf \hat{X}} \times \tilde{X} \times {\bf \hat{U}} \times \tilde{U}$ defined by (\ref{def_bf_R(epsilon)}) is a $(\kappa+\kappa', \max\{ \beta, \beta' \}, \max\{ \lambda, \lambda' \})$-acASR from $\bf \hat{S}$ to $\tilde{S}$ with $\bar{\sf d}$.
\end{lem}
\begin{proof}
We will show that ${\bf R}(\epsilon)$ satisfies the conditions of a $(\kappa+\kappa', \max\{\beta, \beta'\}, \max\{\lambda, \lambda'\})$-acASR from $\bf \hat{S}$ to $\tilde{S}$ with $\bar{\sf d}$.
\begin{enumerate}
\item Consider any $\bf \hat{x}_0 \in \hat{X}_0$.
We consider the following state $\tilde{x}_0$ of $\tilde{S}$:
\begin{align*}
\tilde{x}_0 = \{ \check{x}_0 \in \check{X}_0 \ | \ \exists \hat{x}_0 \in {\bf \hat{x}_0}, \exists x_{0} \in X_0: (\hat{x}_0, x_{0}) \in R_X(\kappa) \land (x_{0}, \check{x}_0) \in \check{R}_X(\kappa') \}.
\end{align*}
Since $R(\epsilon)$ and $\check{R}(\epsilon)$ are a $(\kappa, \beta, \lambda)$-acASR from $\hat{S}$ to $S$ and a $(\kappa', \beta', \lambda')$-acSR from $S$ to $\check{S}$, respectively, $\tilde{x}_0$ is a non-empty set.
Recall that $\tilde{X}_0 = 2^{\check{X}_0} \setminus \{ \emptyset \}$, and we have $\tilde{x}_0 \in \tilde{X}_0$.
By the definition of ${\bf R}(\epsilon)$, $({\bf \hat{x}_0}, \tilde{x}_0) \in {\bf R}_X(\kappa+\kappa')$ holds.
\item First, consider any $({\bf \hat{x}}, \tilde{x}) \in {\bf R}_X(\epsilon)$.
Choose any $\hat{u} \in {\bf \hat{U}}({\bf \hat{x}})$.
From (\ref{assum_R(epsilon)}) and (\ref{assum_Rcheck(epsilon)}), there exist $u \in U$ and $\check{u} \in \tilde{U}(\tilde{x})$ such that $({\bf \hat{x}}, \tilde{x}, \hat{u}, \check{u}) \in {\bf R}(\epsilon)$.
By the definition of ${\bf R}(\epsilon)$, there exist $\epsilon_1 \ge \kappa$ and $\epsilon_2 \ge \kappa'$ such that $\epsilon_1+\epsilon_2=\epsilon$, and the following condition holds:
\begin{align*}
\forall \check{x} \in \tilde{x}, \exists x \in X, \exists \hat{x} \in {\bf \hat{x}}: (\hat{x}, x, \hat{u}, u) \in R(\epsilon_1) \land (x, \check{x}, u, \check{u}) \in \check{R}(\epsilon_2),
\end{align*}
which implies together with the definition of $\tilde{r}$ that
\begin{align*}
\forall \check{x}' \in \bigcup_{\check{x} \in \tilde{x}}{\check{r}(\check{x}, u)}, \exists x' \in \bigcup_{\check{x} \in \tilde{x}, (\check{x}, x) \in \check{R}_X(\epsilon_2)}{r(x, u)}: \ (x', \check{x}') \in \check{R}_X(\kappa'+\beta'\epsilon_2+\lambda'\check{\sf d}(u, \check{u})).
\end{align*}
By the definition of $(\kappa, \beta, \lambda)$-acASR from $\hat{S}$ to $S$ with $\sf d$, we have the following condition:
\begin{align*}
\forall x' \in \bigcup_{\check{x} \in \tilde{x}, (\check{x}, x) \in \check{R}_X(\epsilon_2)}{r(x, u)}, \ \exists \hat{x}' \in \bigcup_{\hat{x} \in {\bf \hat{x}}}{\hat{r}(\hat{x}, \hat{u})}: \ (\hat{x}', x') \in R_X(\kappa+\beta\epsilon_1+\lambda{\sf d}(\hat{u}, u)).
\end{align*}
Thus, we have the following condition:
\begin{align}
\label{eq_pr_2}
\forall \check{x}' \in \bigcup_{\check{x} \in \tilde{x}}{\check{r}(\check{x}, \check{u})}, \ &\exists x' \in \bigcup_{\check{x} \in \tilde{x}, (\check{x}, x) \in \check{R}_X(\epsilon_2)}{r(x, u)}, \ \exists \hat{x}' \in \bigcup_{\hat{x} \in {\bf \hat{x}}}{\hat{r}(\hat{x}, \hat{u})}: \nonumber \\
&(\hat{x}', x') \in R_X(\kappa+\beta\epsilon_1+\lambda{\sf d}(\hat{u}, u)) \land (x', \check{x}') \in \check{R}_X(\kappa'+\beta'\epsilon_2+\lambda'\check{\sf d}(u, \check{u})).
\end{align}
Next, consider any $\tilde{x}' \in \tilde{r}(\tilde{x}, \check{u})$.
By the definition of $\tilde{r}$, we have
\begin{align*}
\tilde{x}' \subseteq \bigcup_{\check{x} \in \tilde{x}}{\check{r}(\check{x}, \check{u})}.
\end{align*}
By the definition of $\bf \hat{r}$, we have
\begin{align*}
{\bf \hat{r}}({\bf \hat{x}}, \hat{u}) = 2^{\bigcup_{\hat{x} \in {\bf \hat{x}}}{\hat{r}(\hat{x}, \hat{u})}} \setminus \{\emptyset\},
\end{align*}
which implies together with (\ref{eq_pr_2}) that there always exists ${\bf \hat{x}'} \in {\bf \hat{r}}({\bf \hat{x}}, \hat{u})$ satisfying the following condition:
\begin{align*}
&\forall \check{x}' \in \tilde{x}', \exists x' \in X, \exists \hat{x}' \in {\bf \hat{x}}': (\hat{x}', x') \in R_X(\kappa+\beta\epsilon_1+\lambda{\sf d}(\hat{u}, u)) \land (x', \check{x}') \in \check{R}_X(\kappa'+\beta'\epsilon_2+\lambda'\check{\sf d}(u, \check{u})).
\end{align*} 
Thus, by the definition of ${\bf R}(\epsilon)$ and (\ref{assum_map}), we have
\begin{align*}
({\bf \hat{x}'}, \tilde{x}') &\in {\bf R}(\kappa+\kappa'+\beta\epsilon_1+\beta'\epsilon_2+\lambda{\sf d}(\hat{u}, u)+\lambda'\check{\sf d}(u, \check{u}))\\
&\subseteq {\bf R}(\kappa+\kappa'+\max\{\beta, \beta'\}\epsilon+\max\{ \lambda, \lambda' \} \bar{\sf d}(\hat{u}, \check{u})).
\end{align*}
 \qed
\end{enumerate}
\end{proof}
By Lemmas \ref{lem_bf hat{R}_C} and \ref{lem_bf_R}, $\bf \hat{R}_C$ and ${\bf R}(\epsilon)$ are an ASR from $\bf \hat{S}_C$ to $\bf \hat{S}$ and a $(\kappa+\kappa', \max\{ \beta, \beta' \}, \max\{ \lambda, \lambda' \})$-acASR from $\bf \hat{S}$ to $\tilde{S}$ with $\bar{\sf d}$, respectively, which implies together with Theorem \ref{tabuada} that Theorem \ref{result_bf_R_C(epsilon)} holds.

Next, we prove Theorem \ref{result_main}.
%
We will show that $\bar{\bf R}_{\bf C}(\epsilon)$ satisfies the conditions of a $(\kappa+\kappa', \max\{\beta, \beta'\},\\ \max\{\lambda, \lambda'\})$-acASR from $\bar{\bf S}_{\bf C}$ to $({S}, Y, H)$ with $\bar{\sf d}_{\sf C}$.
\begin{enumerate}
\item Consider any $(({\bf \hat{x}_{C0}}, {\bf \hat{x}_0}), \tilde{x}_0) \in \bar{\bf X}_{\bf C0}$.
By the definition of the composed systems, we have ${\bf (\hat{x}_{C0}, \hat{x}_0)} \in {\bf \hat{R}}_{{\bf C}X}$ and $({\bf \hat{x}_0}, \tilde{x}_0) \in {\bf R}_X(\kappa+\kappa')$.
By the proof (1) of Lemma \ref{lem_bf_R}, there exists $x_0 \in X_0$ such that
\begin{align*}
\exists \hat{x}_0 \in {\bf \hat{x}_0}, &\exists \check{x}_0 \in \tilde{x}_0: (\hat{x}_0, x_0) \in R_X(\kappa) \land (x_0, \check{x}_0) \in \check{R}_X(\kappa'),
\end{align*}
which implies together with the definition of $R'(\epsilon)$ that we have $((({\bf \hat{x}_{C0}}, {\bf \hat{x}_0}), \tilde{x}_0), x_0) \in {\bf \bar{R}}_{{\bf C}X}(\kappa+\kappa'))$. 
\item First, consider any $((({\bf \hat{x}_C}, {\bf \hat{x}}), \tilde{x}), x) \in {\bf \bar{R}}_{{\bf C}X}(\epsilon)$.
Choose any $((\hat{u}_C, \hat{u}), \check{u}) \in \bar{\bf U}_{\bf C}((({\bf \hat{x}_C}, {\bf \hat{x}}), \tilde{x}))$.
By the definition of the composed systems, we have $({\bf \hat{x}_C}, {\bf \hat{x}}, \hat{u}_C, \hat{u}) \in {\bf \hat{R}}_{\bf C}$ and $({\bf \hat{x}}, \tilde{x}_0, \hat{u}, \check{u}) \in {\bf R}(\epsilon)$.
By the definition of ${\bf R}(\epsilon)$, there exists $u \in U(x)$ such that
\begin{align*} 
\exists \epsilon_1 \ge \kappa, \exists \epsilon_2 \ge \kappa', \exists \check{x} \in {\tilde{x}}, \exists \hat{x} \in {\bf \hat{x}}: \epsilon_1 + \epsilon_2 = \epsilon \land (\hat{x}, x, \hat{u}, u) \in R(\epsilon_1) \land (x, \check{x}, u, \check{u}) \in \check{R}(\epsilon_2),
\end{align*}
which implies together with the definition of $R'(\epsilon)$ that we have $((({\bf \hat{x}_C}, {\bf \hat{x}}), \tilde{x}), x, ((\hat{u}_C, \hat{u}), \check{u}), u) \in {\bf \bar{R}}_{\bf C}(\epsilon)$.

Next, consider any $x' \in r(x, u)$.
Since $R(\epsilon)$ is a $(\kappa, \beta, \lambda)$-acASR from $\hat{S}$ to $S$ with $\sf d$, there exists $\hat{x}' \in \hat{r}(\hat{x}, \hat{u})$ such that $(\hat{x}', x') \in R_X(\kappa + \beta \epsilon_1 + \lambda {\sf d}(\hat{u}, u))$.
Moreover, since $R'(\epsilon)$ is a $(\kappa', \beta', \lambda')$-acSR from $S$ to $\tilde{S}$ with $\check{\sf d}$, there exists $\tilde{x}' \in \tilde{r}(\tilde{x}, \check{u})$ such that
\begin{align}
\label{eq_chk_R}
\exists \check{x}' \in \tilde{x}': (x', \check{x}') \in \check{R}_X(\kappa'+\beta'\epsilon_2+\lambda'\check{\sf d}(u, \check{u})).
\end{align}
Since ${\bf R}(\epsilon)$ is a $(\kappa+\kappa', \max\{\beta, \beta'\}, \max\{\lambda, \lambda'\})$-acASR from $\bf \hat{S}$ to $\tilde{S}$ with $\bar{\sf d}$, there exists ${\bf \hat{x}'} \in {\bf \hat{r}}({\bf \hat{x}}, \hat{u})$ such that
\begin{align*}
\hat{x}' \in {\bf \hat{x}'} \ \land \ ({\bf \hat{x}}', \tilde{x}') \in {\bf R}_X(\kappa + \kappa' + \max \{ \beta, \beta' \} \epsilon + \max \{ \lambda, \lambda' \} \bar{\sf d}(\hat{u}, \check{u})).
\end{align*}
Thus, we have
\begin{align}
\label{eq_hat_dash}
\exists \hat{x}' \in {\bf \hat{x}'}: (\hat{x}', x') \in {R}_X(\kappa + \beta \epsilon_1 + \lambda {\sf d}(\hat{u}, u)).
\end{align}
By the ASR ${\bf \hat{R}_C}$, we have
\begin{align*}
\exists {\bf \hat{x}'_C} \in {\bf \hat{r}_C}({\bf \hat{x}_C}, \hat{u}_C): ({\bf \hat{x}'_C}, {\bf \hat{x}'}) \in {\bf \hat{R}_{CX}}.
\end{align*}
By the definition of the composed systems, we have
\begin{align*}
(({\bf \hat{x}'_C}, {\bf \hat{x}}'), \tilde{x}') \in {\bf \bar{r}_C}((({\bf \hat{x}_C}, {\bf \hat{x}}), \tilde{x}), ((\hat{u}_C, \hat{u}), \check{u})).
\end{align*}
On the other hand, by the definition of $\bar{\sf d}_{\sf C}$ and (\ref{assum_map}), we have
\begin{align*}
 \bar{\sf d}_{\sf C}(((\hat{u}_C, \hat{u}), \check{u}), u) = \bar{\sf d}(\hat{u}, \check{u}) \ge {\sf d}(\hat{u}, u) + \check{\sf d}(u, \check{u}).
\end{align*}
Therefore, by (\ref{eq_chk_R}) and (\ref{eq_hat_dash}), we have
\begin{align*}
((({\bf \hat{x}'_C}, {\bf \hat{x}}'), \tilde{x}'), x') \in {\bf \bar{R}}_{{\bf C}X}(\kappa+\kappa'+\max \{ \beta,  \beta' \} \epsilon + \max \{ \lambda, \lambda' \} \bar{\sf d}_{\sf C}(((\hat{u}_C, \hat{u}), \check{u}), u)).
\end{align*}
\end{enumerate}

\section{Illustrative Example}
\label{Sec_illustrative_example}

\subsection{Physical Plant}

\begin{figure}[b]
\centering
\includegraphics[width=5.5cm,bb=0 0 604 361]{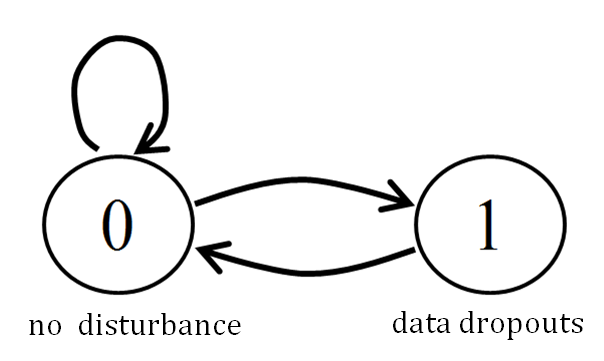}
\caption{The dynamics of the communication channel.}
\label{fig_channel}
\end{figure}
We consider a physical plant given by
\begin{align}
\left\{
\begin{array}{l}
\left[
 \begin{array}{c}
   {\xi}_1[k+1]\\
   {\xi}_2[k+1]
 \end{array}
\right]
=\left[
 \begin{array}{ccc}
   0.5 & 0\\
   0 & 0.25
 \end{array}
\right]\left[
 \begin{array}{c}
   {\xi}_1[k]\\
   {\xi}_2[k]
 \end{array}
\right] + \left[
 \begin{array}{c}
  3.6056\\
  3.9051
 \end{array}
\right] u[k],\\
y_c[k] = \left[
 \begin{array}{ccc}
   2.7042 & 2.2535
 \end{array}
 \right] \left[
 \begin{array}{c}
  {\xi}_1[k]\\
  {\xi}_2[k]
 \end{array}
 \right], \\
y[k] = rd_{\mathbb{Z}} \left( y_c[k]  \right),
\label{eq_contractive}
\end{array}
\right.
\end{align}
where $y_c[k]$ is the output of the plant, and $y[k]$ is the measured value of the sensor.
$rd_{\mathbb{Z}}(y)$ is a rounded value of $y$ to an integer.
We will design a symbolic controller determining the control input $u[k]$. 
Data transmission between the plant and the controller is done via unreliable communication channels where data dropouts sometimes occur.
If data dropouts occur, the control signal $u$ (resp. the output signal $y$) is set to be $0$ at the plant (resp. at the controller).
Assume that data dropouts never occur consecutively.
The dynamics of communication channels is modeled by a system shown in Fig.~\ref{fig_channel}.
First, we introduce a system with outputs $(S, Y, H) = (X, X_{0}, U, r, Y, H)$ that represents the dynamics of the plant and the dynamics of the communication channels, where $X = \mathbb{R}^2 \times \{ 0, 1\} \times \{ 0, 1 \}$, $X_0 = \{ [ 0 \ 0 \ 0 \ 0 ]^T \}$, $U=\{ 0.032, 0.064 \}$, and $Y = \mathbb{Z}$. $H: X \to Y$ is defined as follows: 
\begin{align*}
H(\left[ \ \xi_1[k] \ \xi_2[k] \ \xi_3[k] \ \xi_4[k] \right]^T) = \begin{cases}
rd_{\mathbb{Z}}\left\{ 2.7042 \xi_1[k] + 2.2535 \xi_2[k] \right\} &\mbox{if } \xi_4[k] = 0,\\
0 &\mbox{if } \xi_4[k] = 1.
\end{cases}
\end{align*}
Note that $\xi_3[k]$ (resp. $\xi_4[k]$) is a state of communication channels from the controller to the physical plant (resp. vice versa), and their Boolean values correspond to the state number of Fig.~\ref{fig_channel}.
The transition map $r: X \times U \to 2^{X}$ is defined as follows: 
\begin{align*}
&r(\left[ \ \xi_1 \ \xi_2 \ 0 \ 0 \right]^T, u) = \left\{  \left[ \ \xi'_1 \ \xi'_2 \ 0 \ 0 \right]^T,  \left[ \ \xi'_1 \ \xi'_2 \ 0 \ 1 \right]^T, \left[ \ \xi''_1 \ \xi''_2 \ 1 \ 0 \right]^T, \left[ \ \xi''_1 \ \xi''_2 \ 1 \ 1 \right]^T  \right\},\\
&r(\left[ \ \xi_1 \ \xi_2 \ 0 \ 1 \right]^T, u) = \left\{  \left[ \ \xi'_1 \ \xi'_2 \ 0 \ 0 \right]^T, \left[ \ \xi''_1 \ \xi''_2 \ 1 \ 0 \right]^T \right\},\\
&r(\left[ \ \xi_1 \ \xi_2 \ 1 \ 0 \ \right]^T, u) = \left\{  \left[ \ \xi'_1 \ \xi'_2 \ 0 \ 0 \right]^T, \left[ \ \xi'_1 \ \xi'_2 \ 0 \ 1 \right]^T \right\}, \ \ r(\left[ \ \xi_1 \ \xi_2 \ 1 \ 1 \right]^T, u) = \left\{  \left[ \ \xi'_1 \ \xi'_2 \ 0 \ 0 \right]^T \right\},
\end{align*}
where
\begin{align}
\left[
 \begin{array}{c}
   {\xi}'_1\\
   {\xi}'_2
 \end{array}
\right] 
=
\begin{array}{l}
\left[
 \begin{array}{ccc}
   0.5 & 0\\
   0 & 0.25
 \end{array}
\right]\left[
 \begin{array}{c}
   {\xi}_1\\
   {\xi}_2
 \end{array}
\right] + \left[
 \begin{array}{c}
  3.6056\\
  3.9051
 \end{array}
\right]u , \ \ 
\left[
 \begin{array}{c}
   {\xi}''_1\\
   {\xi}''_2
 \end{array}
\right] 
=\left[
 \begin{array}{ccc}
   0.5 & 0\\
   0 & 0.25
 \end{array}
\right]\left[
 \begin{array}{c}
   {\xi}_1\\
   {\xi}_2
 \end{array}
\right].
\label{eq_transition_relation}
\end{array}
\end{align}
Second, let $[A]_{\eta} := \{ x \in A \ | \ \exists k \in \mathbb{Z}^n : x = 2 k \eta \}$ for a set $A \subseteq \mathbb{R}^{n}$.
Then, we introduce a c-abstracted system $\hat{S} = (\hat{X}, \hat{X}_0, \hat{U}, \hat{r})$ where $\hat{X} = [[0, 0.4]^2]_{0.005} \times \{ 0, 1 \} \times \{ 0, 1 \}$, $\hat{X}_0 = \{ [ 0 \ 0 \ 0 \ 0 ]^T \}$, and $\hat{U}=\{ 0.032, 0.064 \}$. 
The transition map $\hat{r}: \hat{X} \times \hat{U} \to 2^{\hat{X}}$is defined as follows: 
\begin{align*}
&\hat{r}(\left[ \ \hat{\xi}_1 \ \hat{\xi}_2 \ 0 \ 0 \right]^T, \hat{u}) = \left\{  \left[ \ \hat{\xi}'_1 \ \hat{\xi}'_2 \ 0 \ 0 \right]^T,  \left[ \ \hat{\xi}'_1 \ \hat{\xi}'_2 \ 0 \ 1 \right]^T, \left[ \ \hat{\xi}''_1 \ \hat{\xi}''_2 \ 1 \ 0 \right]^T, \left[ \ \hat{\xi}''_1 \ \hat{\xi}''_2 \ 1 \ 1 \right]^T  \right\},\\
&\hat{r}(\left[ \ \hat{\xi}_1 \ \hat{\xi}_2 \ 0 \ 1 \right]^T, \hat{u}) = \left\{  \left[ \ \hat{\xi}'_1 \ \hat{\xi}'_2 \ 0 \ 0 \right]^T, \left[ \ \hat{\xi}''_1 \ \hat{\xi}''_2 \ 1 \ 0 \right]^T \right\},\\
&\hat{r}(\left[ \ \hat{\xi}_1 \ \hat{\xi}_2 \ 1 \ 0 \ \right]^T, \hat{u}) = \left\{  \left[ \ \hat{\xi}'_1 \ \hat{\xi}'_2 \ 0 \ 0 \right]^T, \left[ \ \hat{\xi}'_1 \ \hat{\xi}'_2 \ 0 \ 1 \right]^T \right\}, \ \ 
\hat{r}(\left[ \ \hat{\xi}_1 \ \hat{\xi}_2 \ 1 \ 1 \right]^T, \hat{u}) = \left\{  \left[ \ \hat{\xi}'_1 \ \hat{\xi}'_2 \ 0 \ 0 \right]^T \right\},
\end{align*}
where
\begin{align}
\left[
 \begin{array}{c}
   \hat{\xi}'_1\\
   \hat{\xi}'_2
 \end{array}
\right] 
= rd_2 \Biggl(
\begin{array}{l}
\left[
 \begin{array}{ccc}
   0.5 & 0\\
   0 & 0.25
 \end{array}
\right]\left[
 \begin{array}{c}
   \hat{\xi}_1\\
   \hat{\xi}_2
 \end{array}
\right] + \left[
 \begin{array}{c}
  3.6056\\
  3.9051
 \end{array}
\right]\hat{u} \Biggr), \ \ 
\left[
 \begin{array}{c}
   \hat{\xi}''_1\\
   \hat{\xi}''_2
 \end{array}
\right] 
= rd_2 \Biggl( 
\left[
 \begin{array}{ccc}
   0.5 & 0\\
   0 & 0.25
 \end{array}
\right]\left[
 \begin{array}{c}
   \hat{\xi}_1\\
   \hat{\xi}_2
 \end{array}
\right]\Biggr).
\label{eq_transition_relation}
\end{array}
\end{align}
Note that $rd_{2}(y)$ rounds $y$ off to two decimal place. 
Then, the following relation ${R}(\epsilon) \subseteq \hat{X} \times {X} \times \hat{U} \times {U}$ is a $(0.005, 0.5, 0)$-acASR from $\hat{S}$ to ${S}$:
\begin{align}
\label{R(epsilon)_ex}
{R}(\epsilon) = \left\{ \left. \left(
\left[
 \begin{array}{c}
   \hat{\xi}_1\\
   \hat{\xi}_2\\
   \hat{\xi}_3\\
   \hat{\xi_4}
 \end{array}
\right],
\left[
 \begin{array}{c}
   {\xi}_1\\
   {\xi}_2\\
   {\xi}_3\\
   {\xi}_4
 \end{array}
\right],
\hat{u},
{u}
\right) \ \right| \ \hat{u} = {u} \ \land \ 
\left|
\left[
 \begin{array}{c}
   \hat{\xi}_1\\
   \hat{\xi}_2
 \end{array}
\right]-\left[
 \begin{array}{c}
   {\xi}_1\\
   {\xi}_2
 \end{array}
\right] 
\right| \le \epsilon \ \land \ \hat{\xi}_3 = \xi_3 \ \land \ \hat{\xi}_4 = \xi_4
\right\}.
\end{align}

The desired behavior is that $y_c[k]$ converges to $1$.
Then, we introduce $\hat{S}_C = (\hat{X}_C, \hat{X}_{C0}, \hat{U}_C, \hat{r}_C)$ as shown in Fig.~\ref{fig_SChat} where $\hat{X}_{C0} = \left\{ [ 0 \ 0 \ 0 \ 0 ]^T \right\}$ and $\hat{U}_C = \{ 0.032, 0.064 \}$.
The red arrows in Fig.~\ref{fig_SChat} describe the transitions when $\hat{u}_C = 0.064$, and the black arrows describe the transitions when $\hat{u}_C = 0.032$.

Then, the following relation $\hat{R}_C \subseteq \hat{X}_C \times \hat{X} \times \hat{U}_C, \hat{U}$ is an ASR from $\hat{S}_C$ to $\hat{S}$:
\begin{align}
\label{RChat_ex}
\hat{R}_C = \{ (\hat{x}_C, \hat{x}, \hat{u}_C, \hat{u}) \ | \ \hat{x}_C = \hat{x} \ \land \ \hat{u}_C = \hat{u}  \}.
\end{align}
Third, we introduce an o-abstracted system $\check{S} = (\check{X}, \check{X}_0, \check{U}, \check{r})$ where $\check{X} = [[0, 0.4]^2]_{0.05} \times \{ 0, 1 \} \times \{ 0, 1 \}$, $\check{X}_0 = \{ [0 \ 0 \ 0 \ 0]^T \}$, and $\check{U}=\{ 0.032, 0.064 \}$.
The transition map $\check{r}: \check{X} \times \check{U} \to 2^{\check{X}}$is defined as follows: 
\begin{align*}
&\check{r}(\left[ \ \check{\xi}_1 \ \check{\xi}_2 \ 0 \ 0 \right]^T, \check{u}) = \left\{  \left[ \ \check{\xi}'_1 \ \check{\xi}'_2 \ 0 \ 0 \right]^T,  \left[ \ \check{\xi}'_1 \ \check{\xi}'_2 \ 0 \ 1 \right]^T, \left[ \ \check{\xi}''_1 \ \check{\xi}''_2 \ 1 \ 0 \right]^T, \left[ \ \check{\xi}''_1 \ \check{\xi}''_2 \ 1 \ 1 \right]^T  \right\},\\
&\check{r}(\left[ \ \check{\xi}_1 \ \check{\xi}_2 \ 0 \ 1 \right]^T, \check{u}) = \left\{  \left[ \ \check{\xi}'_1 \ \check{\xi}'_2 \ 0 \ 0 \right]^T, \left[ \ \check{\xi}''_1 \ \check{\xi}''_2 \ 1 \ 0 \right]^T \right\},\\
&\check{r}(\left[ \ \check{\xi}_1 \ \check{\xi}_2 \ 1 \ 0 \ \right]^T, \check{u}) = \left\{  \left[ \ \check{\xi}'_1 \ \check{\xi}'_2 \ 0 \ 0 \right]^T, \left[ \ \check{\xi}'_1 \ \check{\xi}'_2 \ 0 \ 1 \right]^T \right\}, \ \ 
\check{r}(\left[ \ \check{\xi}_1 \ \check{\xi}_2 \ 1 \ 1 \right]^T, \check{u}) = \left\{  \left[ \ \check{\xi}'_1 \ \check{\xi}'_2 \ 0 \ 0 \right]^T \right\},
\end{align*}
where
\begin{align}
\left[
 \begin{array}{c}
   \check{\xi}'_1\\
   \check{\xi}'_2
 \end{array}
\right] 
= rd_1 \Biggl(
\begin{array}{l}
\left[
 \begin{array}{ccc}
   0.5 & 0\\
   0 & 0.25
 \end{array}
\right]\left[
 \begin{array}{c}
   \check{\xi}_1\\
   \check{\xi}_2
 \end{array}
\right] + \left[
 \begin{array}{c}
  3.6056\\
  3.9051
 \end{array}
\right]\check{u} \Biggr), \ \ 
\left[
 \begin{array}{c}
   \check{\xi}''_1\\
   \check{\xi}''_2
 \end{array}
\right] 
= rd_1 \Biggl( 
\left[
 \begin{array}{ccc}
   0.5 & 0\\
   0 & 0.25
 \end{array}
\right]\left[
 \begin{array}{c}
   \check{\xi}_1\\
   \check{\xi}_2
 \end{array}
\right]\Biggr).
\label{eq_transition_relation}
\end{array}
\end{align}
Note that $rd_{1}(y)$ rounds $y$ off to one decimal place. 

Then, the following relation $\check{R}(\epsilon) \subseteq {X} \times \check{X} \times {U} \times \check{U}$ is a $(0.05, 0.5, 0)$-acSR from ${S}$ to $\check{S}$:
\begin{align}
\label{checkR(epsilon)_ex}
\check{R}(\epsilon) = \left\{ \left. \left(
\left[
 \begin{array}{c}
   {\xi}_1\\
   {\xi}_2\\
   {\xi}_3\\
   {\xi}_4
 \end{array}
\right],
\left[
 \begin{array}{c}
   \check{\xi}_1\\
   \check{\xi}_2\\
   \check{\xi}_3\\
   \check{\xi}_4
 \end{array}
\right],
u,
\check{u}
\right) \ \right| \ {u} = \check{u} \ \land \ 
\left|
\left[
 \begin{array}{c}
   {\xi}_1\\
   {\xi}_2
 \end{array}
\right]-\left[
 \begin{array}{c}
   \check{\xi}_1\\
   \check{\xi}_2
 \end{array}
\right] 
\right| \le \epsilon \ \land \ {\xi}_3 = \check{\xi}_3 \ \land \ {\xi}_4 = \check{\xi}_4
\right\}.
\end{align}

\subsection{Output Feedback Control}

We construct an observer $\tilde{S}=(\tilde{X}, \tilde{X}_0, \tilde{U}, \tilde{r})$ induced by $\check{S}$ as shown in Definition \ref{def_observer}. 
Consider the ASR $\hat{R}_C$, $(0.005, 0.5, 0)$-acASR $R(\epsilon)$, and $(0.05, 0.5, 0)$-acSR $\check{R}(\epsilon)$ defined in the previous subsection.
Then, it is shown that $\hat{R}_C$ satisfies (\ref{restrictive_R_Chat}), that $R(\epsilon)$ satisfies (\ref{assum_R(epsilon)}), and that $\check{R}(\epsilon)$ satisfies (\ref{assum_Rcheck(epsilon)}).
Note that $U=\hat{U}=\check{U}$.
We introduce ${\bf \hat{S}}=({\bf \hat{X}}, {\bf \hat{X}_0}, {\bf \hat{U}}, {\bf \hat{r}})$ induced by $\hat{S}$ and ${\bf \hat{S}_C}=({\bf \hat{X}_C}, {\bf \hat{X}_{C0}}, {\bf \hat{U}_C}, {\bf \hat{r}_C})$ induced by $\hat{S}_C$ defined in Definitions \ref{def_bf_Shat} and \ref{def_bf_S_Chat}, respectively.
Now, we use Theorems \ref{result_bf_R_C(epsilon)} and \ref{result_main} to design the output feedback controller $\bar{\bf S}_{\sf C}$, where $\bar{\bf R}_{\bf C}(\epsilon)$ is a $(0.055, 0.5, 0)$-acASR from $\bar{\bf S}_{\bf C}$ to $S$.

\subsection{Simulation Result}

By the computer simulation, it is shown that the observer $\tilde{S}$ has 10 states, and that the output feedback controller $\bar{\bf S}_{\bf C}$ has 27 states.
The occurrences of data dropouts are shown in Fig.~\ref{fig_disturbance}.
The time response of $y_c[k]$ is shown in Fig.~\ref{fig_position}.
It is shown that $y_c[k]$ converges to $1$ though sometimes deviates by the input disturbances.
The numbers of candidates of the current state are shown in Fig.~\ref{fig_number}.
Since we have $\check{X} \subseteq \hat{X}$, it is noticed that the number of candidates listed up by the controller is always larger than that by the observer.

\begin{figure}
\begin{minipage}{0.5\hsize}
\centering
\includegraphics[width=\hsize]{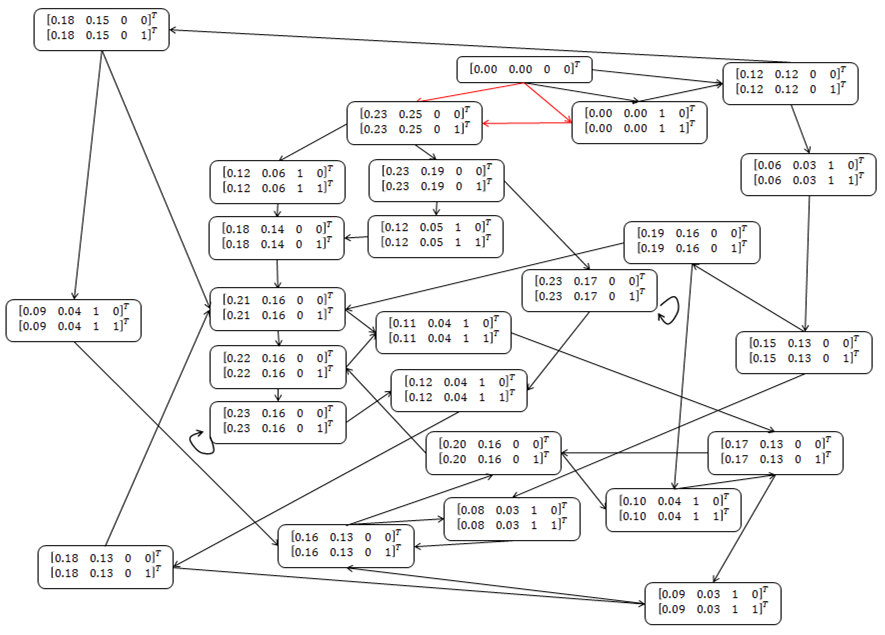}
\caption{The behavior of $\hat{S}_C$.}
\label{fig_SChat}
\end{minipage}
\begin{minipage}{0.5\hsize}
\centering
\includegraphics[width=\hsize]{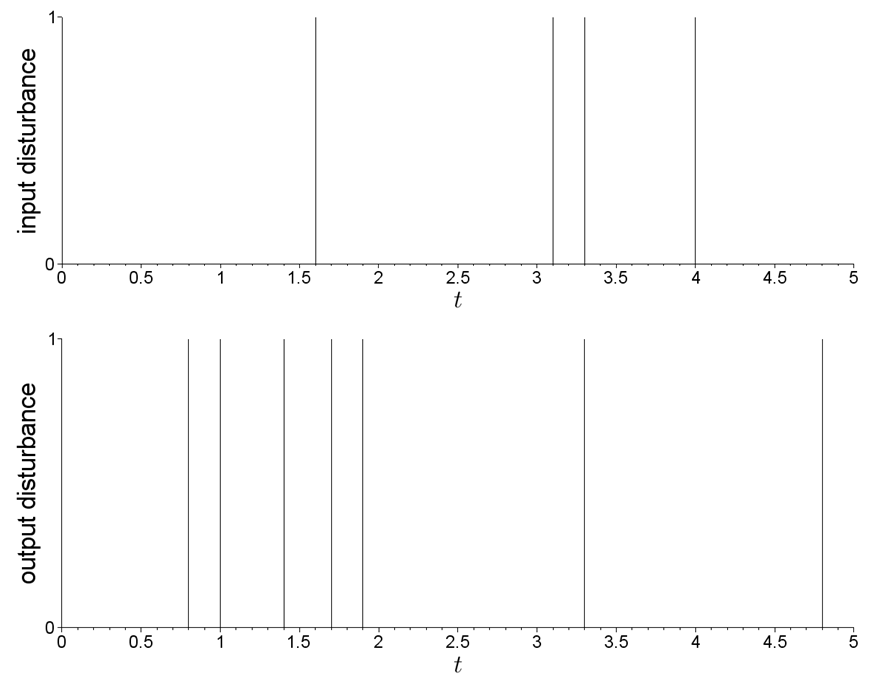}
\caption{The occurrences of data dropouts.}
\label{fig_disturbance}
\end{minipage}
\end{figure}

\begin{figure}
\begin{minipage}{0.5\hsize}
\centering
\includegraphics[width=\hsize]{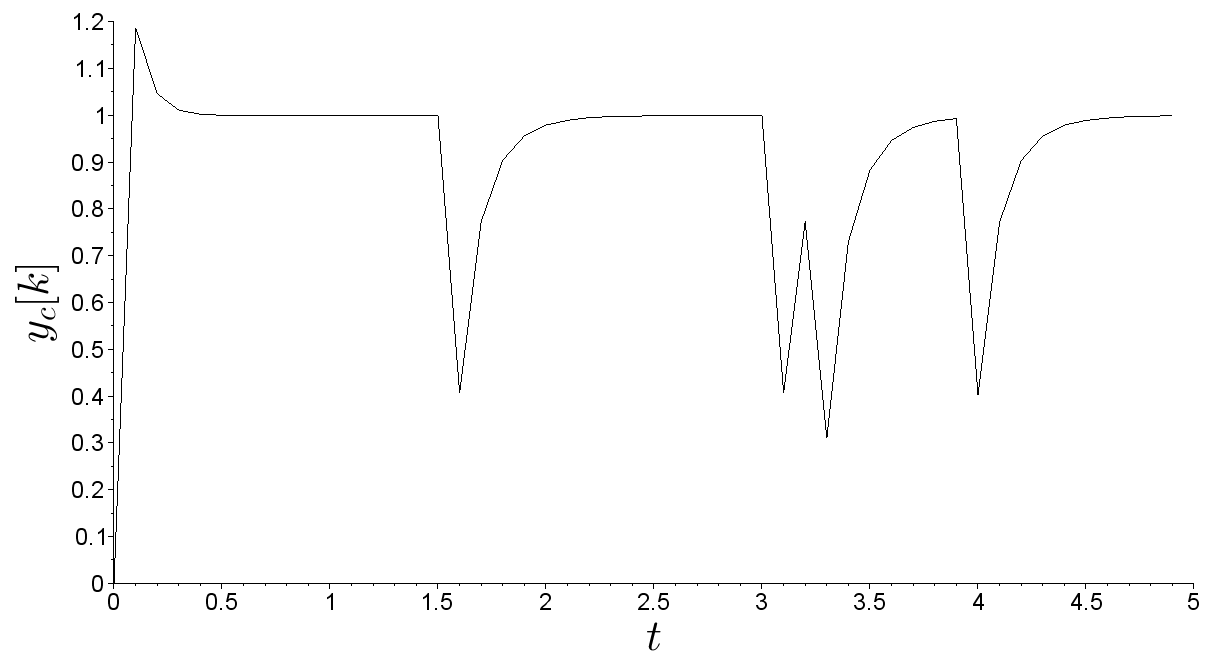}
\caption{The time response of $y_c[k]$.}
\label{fig_position}
\end{minipage}
\begin{minipage}{0.5\hsize}
\centering
\includegraphics[width=\hsize]{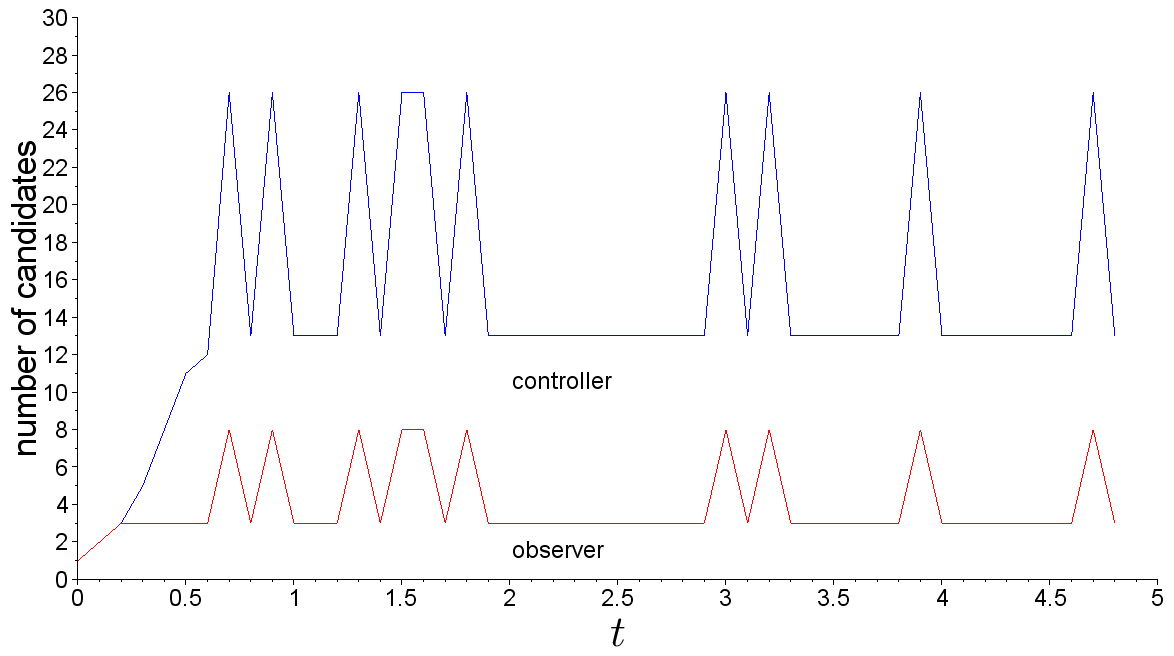}
\caption{The numbers of candidates of the current state.}
\label{fig_number}
\end{minipage}
\end{figure}

\section{Conclusion}

We consider a symbolic design of an observer that lists up all possible candidates of the current state of the physical plant.
In order to design a symbolic output feedback controller, two abstracted systems are introduced.
One abstracted system needs an acASR to construct a feedback controller.
The other needs an acSR to estimate the state of the plant.
They are given independently, which means that the separation principle of the control and the observation holds. 
We proved that there exists an acASR from the proposed output feedback controller to the physical plant without introducing the distance.

It is shown in \cite{6760490, rungger2014abstracting, 2013arXiv1310.5199R, 6882838} that the input-output dynamical stability is preserved under an acASR.
Thus, it is future work to show the stability of the controlled plant. 

\section*{Acknowledgment}
This work was supported by JSPS KAKENHI No.~15K14007.

\bibliographystyle{eptcs}
\bibliography{generic2}
\end{document}